\documentclass{article}
\usepackage[alg]{klmm}

\newcommand{\F}{\mathbb{F}}

\title{Deterministic Polynomial-time Exact-root Computation  for Sparse Polynomials with Bounded Total Degree} 
\author{
    Qiao-Long Huang\institution{Shandong University, School of Mathematics,
Jinan, China}
    \and Yichuan Cao\institution{State Key Laboratory of Mathematical Sciences, Academy of Mathematics and Systems Science, Chinese Academy of Sciences; University of Chinese Academy of Sciences, Beijing China}
    \and Ruichen Qiu\instref{2}
    \and Xiao-Shan Gao\instref{2}
}
\date{\today}

\newcommand{\supp}{\operatorname{supp}}
\newcommand{\indeg}{\operatorname{indeg}}
\newcommand{\poly}{\operatorname{poly}}

\newcommand{\tdeg}{\operatorname{tdeg}}

\newcommand{\Dlin}{D_{\mathrm{lin}}}

\numberwithin{theorem}{section}

\begin{document}
\maketitle

\begin{abstract}
We study the problem of deterministically computing the exact root of a sparse polynomial in the multivariate setting. Let $f \in \F[x_1,\ldots,x_n]$ be a nonzero polynomial that is an exact $e$-th power, say $f = g^e$. Suppose $f$ is $s$-sparse, has an individual degree of at most $d$, and a total degree of $D = \tdeg(f)$. We prove a sparsity bound on the base polynomial $g$:
\[
\|g\|_0 \le s^{D(2d+2)/e + 1}.
\]
Based on this bound, we develop a deterministic algorithm that computes the base $g$.
In contrast to the general deterministic factorization algorithm of Bhargava, Saraf, and Volkovich \cite{BhargavaSarafVolkovich2020}, which achieves only a quasi-polynomial dependence on the input parameters, our algorithm is \emph{polynomial-time} in the setting where the total degree $D$ is bounded. Specifically,  the overall complexity is
\[
\mathrm{poly}\left(s^{O(Dd)}, n, d, D\right) + s\cdot R(e),
\]
where $R(e)$ denotes the cost of constructing a single $e$-th root of a scalar in the base field $\F$, and, when $\operatorname{char}(\F)\mid e$, the cost of computing a single Frobenius root of a scalar.
This term is field-dependent, and over finite fields, $\mathbb{Q}$, or number fields with a suitable representation, it is absorbed into the polynomial complexity bound. 
Within the bounded total-degree regime, this yields a deterministic polynomial-time algorithm for exact-root computation.
\end{abstract}

\section{Introduction}

A polynomial is called $s$-sparse if it has at most $s$ nonzero monomials. The sparse (or lacunary) representation of a multivariate polynomial can be exponentially more compact than the dense representation: a polynomial such as $x^N - 1$ requires only $O(\log N)$ bits to specify, while its dense representation has a length of $O(N)$. This extreme compression makes sparse polynomials attractive in algebraic complexity theory, coding theory, and symbolic computation, but it also poses fundamental challenges for deterministic algorithm design. In particular, algorithms whose complexity depends polynomially on the total degree $D$ become exponential in the input size when the polynomial is given sparsely, since $D$ can be exponential in the representation length.

One of the central problems in this area is sparse polynomial factorization: given a sparse polynomial $f$, compute its irreducible factors in a representation that is also concise. Bhargava, Saraf, and Volkovich~\cite{BSV2020} gave a deterministic factorization algorithm for sparse polynomials with bounded individual degree. Their framework establishes a general factor-sparsity bound: if an $s$-sparse polynomial has an individual degree of at most $d$, then every factor has sparsity bounded by a function of $s$ and $d$, enabling deterministic reconstruction. However, their algorithm achieves only quasi-polynomial dependence on the sparse parameters, leaving room for improvement in special but important cases.

A particularly fundamental special case is the computation of exact powers: given a polynomial $f$, determine whether $f = g^e$ for some polynomial $g$ and integer $e \ge 2$, and if so, compute $g$. This problem arises as a subroutine in complete factorization algorithms, in square-free decomposition, and in detecting prime powers prior to integer factorization. For dense polynomials, exact-power testing and root extraction are classical problems solved efficiently via square-free decomposition or Newton iteration. In the sparse setting, however, these methods break down because they depend polynomially on the degree $D$, which may be exponential in the input size.

Bisht and Volkovich~\cite{BishtVolkovich2025} made significant progress on the structural side, proving a sparsity bound for exact roots under a reverse-monic condition: if an $s$-sparse reverse-monic polynomial $H$ with individual degree at most $\Delta$ satisfies $H = R^e$, then $\|R\|_0 \le s^{\Delta/e + 1}$. This result provides a key structural insight, but it applies only to polynomials that are already in reverse-monic form and does not, by itself, yield a complete deterministic algorithm for arbitrary sparse inputs. Giesbrecht and Roche~\cite{GR2011} gave randomized algorithms for detecting lacunary perfect powers and computing their roots, with polynomial-time Monte Carlo detection and root computation that is either output-sensitive or conditional on a conjecture of Schinzel. Their work demonstrates the practicality of sparse perfect-power algorithms but leaves open the question of deterministic, unconditional polynomial-time computation.

In this paper, we address this gap by providing a complete deterministic treatment of the exact-power problem for sparse polynomials. Our main contributions are as follows.

First, we prove a sharp sparsity bound for the base polynomial $g$ in an exact power $f = g^e$. Let $f \in \mathbb{F}[x_1,\ldots,x_n]$ be nonzero $s$-sparse with $\operatorname{indeg}(f) \le d$ and total degree $D = \tdeg(f)$. If $f = g^e$, then
\[
\|g\|_0 \;\le\; s^{\,D(2d+2)/e \,+\, 1}.
\]
The proof proceeds by applying a scalar weighted substitution to transform an arbitrary exact power into a pseudo-monic form in a new variable $y$, to which the pseudo-monic exact-root lemma of Bisht and Volkovich applies. The resulting bound depends only on $s$, $D$, $d$, and $e$, and is independent of the number of variables $n$.

Second, we develop a deterministic algorithm that computes the base $g$ when $f = g^e$, and correctly reports failure otherwise. The algorithm uses a bounded separating vector to isolate a unique lowest monomial of $f$, then constructs a pseudo-monic transformed polynomial $F^\sharp = (G^\sharp)^e$ via the substitution $x_i \mapsto x_i y^{k_i}$. The normalized root is recovered by a $y$-adic binomial expansion; in positive characteristic, where the characteristic divides $e$, we apply a true inverse Frobenius root reduction before the binomial step. After undoing the normalization and substituting $y=1$, the candidate root is verified by the exact equality $g_{\text{cand}}^e = f$. The overall complexity is
\[
\operatorname{poly}\!\left(s^{O(Dd)},\, n,\, d,\, D\right) \;+\; R(e),
\]
where $R(e)$ encapsulates the field-dependent cost of testing and computing $e$-th roots of scalars in the base field $\mathbb{F}$. Over finite fields, $\mathbb{Q}$, or number fields with suitable representation, $R(e)$ is absorbed into the polynomial complexity bound. In the bounded-total-degree regime, this yields a deterministic polynomial-time algorithm for exact-root computation, contrasting with the quasi-polynomial dependence of the general BSV factorization framework.

Third, we prove a sparsity bound for multilinear cofactors in an oriented factorization $f = g h$. If $h$ is multilinear, $\operatorname{indeg}(f) \le d$, and $D = \tdeg(f)$, then
\[
\|h\|_0 \le s, \qquad \|g\|_0 \le s^{(2d+2)D + 2}.
\]
This bound is obtained by the same weighted-substitution technique, combined with the pseudo-monic cofactor bound of Bisht and Volkovich. It provides structural information for factorization problems involving multilinear divisors, complementing the exact-power results.

\subsection{Our Contributions}

The first result is a exact-power sparsity bound. For $f=g^e$ with $\|f\|_0\leq s$, $\indeg(f)\leq d$, and $D=\tdeg(f)$, then
\begin{equation}
\label{eq-bd1}
      \|g\|_0\leq s^{D(2d+2)/e+1}.
\end{equation}

The proof proceeds by transforming the input into a pseudo-monic form via a scalar weighted substitution, to which the pseudo-monic exact-root lemma is applied.

The second result is a deterministic algorithm for computing the $e$-th root of a sparse polynomial. The input consists of a nonconstant polynomial $f\in \F[x_1,\ldots,x_n]$ and an integer $e\ge 2$. If there exists $g\in \F[x_1,\ldots,x_n]$ such that $f=g^e$, the algorithm outputs such a $g$; otherwise, it outputs ``not an $e$-th power.'' The algorithm proceeds via a scalar pseudo-monic transformation, followed by a $y$-adic binomial expansion with Frobenius reduction in positive characteristic, and concludes with an exact verification step. In contrast to the general deterministic factorization algorithm of Bhargava, Saraf, and Volkovich, which achieves only quasi-polynomial dependence on the input parameters, our algorithm is polynomial-time in the setting where the total degree $D$ is bounded. Specifically, the overall complexity is
\[
  \mathrm{poly}\bigl(s^{O(Dd)}, n, d, D\bigr) + s\cdot R(e),
\]
where $R(e)$ denotes the cost of constructing a single $e$-th root of a scalar in the base field $\F$, and, when $\operatorname{char}(\F)\mid e$, the cost of computing a single Frobenius root of a scalar.
Over finite fields, $\mathbb{Q}$, or number fields with a suitable representation, this term is absorbed into the polynomial complexity bound.

The third result is a source-supported multilinear-cofactor bound. For a nonzero oriented factorization $f=gh$ with $\operatorname{indeg}(f)\leq d$, $D=\tdeg(f)$, and $h$ multilinear, we have 
\begin{equation}
\label{eq-bd2}
\|h\|_0\leq s \quad\hbox{\rm and}\quad  \|g\|_0\leq s^{(2d+2)D+2}.
\end{equation}

The main results presented in this paper are obtained through an interaction between the authors and an artificial intelligence agent system, \textit{MechMath Agent Team (MMAT)}~\cite{MMAT}.
The authors assume full responsibility for the paper’s content.
Example
\subsection{Example}

The following example illustrates that the sparsity bound must depend only on the degree and not on the number of variables.

\begin{example}
    \label{thm:q1-proof}
For every \(n\geq 1\) and \(r\geq 1\), there are explicit polynomials
\[
  g,f\in \mathbb F_2[y,x_1,\ldots,x_n]
\]
and an exponent \(e=3\) such that \(f=g^e\),
\[
  \|f\|_0=s=(n+1)^2,
  \qquad
  \operatorname{inddeg}(f)=d=2^{2r-1}+1,
\]
and
\[
  \|g\|_0=s^{(\log_2(d-1)+1)/4}.
\]
\end{example}

\paragraph{Construction.}
Let
\[
  M=2r-1,\qquad q=2^M,\qquad D=\frac{q+1}{3},\qquad
  H=x_1+\cdots+x_n,\qquad z=yH.
\]
The binary identity used in the proof is
\[
  D=\sum_{b\in B}2^b,\qquad
  B=\{0\}\cup\{\,2a-1:1\leq a\leq r-1\,\}.
\]
Thus the nonzero binary positions of \(D\) are
\[
  0,1,3,5,\ldots,2r-3
\]
with only position \(0\) when \(r=1\), and \(3D=q+1\).  Define
\[
  g=(1+z)^D,\qquad f=g^3.
\]
Then \(f=(1+z)^{q+1}\), so the exact-root relation holds.  Since \(z=0\) when
\(y=0\), the pivot specialization is \(f|_{y=0}=1\).

The support count for \(f\) is explicit:
\[
  f=(1+z)^{q+1}=(1+z^q)(1+z)
    =1+z+z^q+z^{q+1}.
\]
The four terms have distinct \(y\)-degrees.  The last term contributes the
\(n^2\) distinct monomials \(y^{q+1}x_i^q x_j\), and the first three
contribute \(1,n,n\) monomials, respectively.  Hence
\[
  \|f\|_0=1+n+n+n^2=(n+1)^2.
\]
For \(g\), the Frobenius expansion gives
\[
  g=\prod_{b\in B}(1+z^{2^b})
   =\sum_{S\subseteq B} z^{\sum_{b\in S}2^b}.
\]
Distinct subsets give distinct \(y\)-degrees.  For each fixed \(S\), binary
digits of the exponents record which variable was chosen from each factor
\(H^{2^b}\), so the fixed-\(S\) contribution has \(n^{|S|}\) monomials.
Therefore
\[
  \|g\|_0=\sum_{S\subseteq B} n^{|S|}=(n+1)^r.
\]
The individual degree of \(f\) is exactly \(q+1=d\), as witnessed by the
noncancelled diagonal monomials \(y^{q+1}x_i^{q+1}\).  Combining
\[
  s=(n+1)^2,\qquad \|g\|_0=(n+1)^r,\qquad d-1=2^{2r-1}
\]
It gives the displayed formula for \(\|g\|_0\). 

This example demonstrates that the root sparsity can indeed grow as \(s^{\Theta(\log d)}\), where the exponent depends only on the individual degree \(d\) and not on the number of variables \(n\). Consequently, any meaningful sparsity bound for exact roots must avoid dependence on \(n\) in the exponent of \(s\), which is precisely the guaranty provided by our theorem.

\subsection{Technical Overview}

Reverse-monicity is the central simplifying condition. If a polynomial $H$ is reverse-monic in a variable $x_j$, then $H|_{x_j=0}=1$. For exact roots, this makes the formal expression
\[
  R=H^{1/e}=(1+(H-1))^{1/e}
\]
truncate in bounded $z_j$-degree, leading to the bound $\|R\|_0\leq S^{\Delta/e+1}$ for $S$-sparse $H$ of individual degree at most $\Delta$ \cite[Lemma 5.3]{BishtVolkovich2025}.

For arbitrary exact powers, the scalar substitution
\[
  f(x_1y^{k_1},\ldots,x_ny^{k_n})
\]
has transformed $y$-degree at most $D(2d+2)$: a monomial $x_1^{e_1}\cdots x_n^{e_n}$ contributes $e_1k_1+\cdots+e_nk_n\leq D(2d+2)$. After a bounded weight vector $k\in\{1,\ldots,2d+2\}^n$ isolates a unique lowest monomial of $f$, we divide by the minimum weight and use
\[
  F^\sharp=y^{-k_f}f(x_1y^{k_1},\ldots,x_ny^{k_n}).
\]
This creates a $\{y\}$-reverse pseudo-monic exact power $F^\sharp=(G^\sharp)^e$ whose $y$-degree is controlled directly by the scalar estimate. A local constants-explicit pseudo-monic exact-root lemma then gives the exact exponent. The earlier constant-free reverse-monic normalization is not used with the $D(2d+2)$ parameter.

\subsection{Related Work}

\paragraph{Sparsity of Exact-root of Sparse Polynomial.} 
If $f=g^e$ for $e\in\mathbb{Z}_{>1}$, $g$ is called an exact-root of $f$.

\textit{The Erdős--Rényi conjecture on sparse powers.}
The study of sparsity behavior under polynomial powers originates from a conjecture of R\'enyi and independently Erd\H{o}s~\cite{Erdos1949}, who conjectured that a bound on the number of terms of $f(x)^2$ implies a bound on the number of terms of $f(x)$ itself. This conjecture was proved by Schinzel~\cite{Schinzel1987} for all powers $f(x)^l$, establishing an explicit lower bound $t \ge c_l \log \log T$, where $T$ and $t$ are the numbers of terms of $f$ and $f^l$, respectively. The bound was later sharpened by Schinzel and Zannier~\cite{SZ2009} to $t \ge 2 + \log(T-1)/\log(4l)$. 

\textit{Schinzel's lower bound for terms of powers.}
In his 1987 paper~\cite{Schinzel1987}, Schinzel proved the Erd\H{o}s--R\'enyi conjecture in full generality, showing that if $f(x) \in K[x]$ (with $\operatorname{char} K = 0$ or $\operatorname{char} K > l \deg f$) has $T$ terms and $f(x)^l$ has $t$ terms, then $t \ge c_l \log \log T$ for an explicit $c_l > 0$. The proof uses a delicate induction on the number of terms, Dirichlet approximation, and properties of Wronskians. This result is a qualitative statement about the unavoidable growth of term count when taking powers. In contrast, our Theorem~3.4 gives an upper bound $\|g\|_0 \le s^{D(2d+2)/e + 1}$ for the base $g$ when $f = g^e$, which is algorithmic in nature and directly enables the deterministic reconstruction of $g$. The combination of Schinzel's lower bounds and our upper bounds delineates the precise sparsity landscape for powers of sparse polynomials.

Schinzel and Zannier~\cite{SZ2009} established lower bounds on the number of terms of powers of polynomials: if $f$ has $T$ terms and $f^l$ has $t$ terms, then $t \ge 2 + \log(T-1)/\log(4l)$ (in characteristic zero or large characteristic). This result, improving earlier work of Schinzel~\cite{Schinzel1987}, is structurally complementary to our upper bound on the sparsity of the root. Together, these results give a more complete picture of the sparsity behavior of powers and roots of polynomials.

%\textbf{Perfect power detection and root computation in lacunary polynomials.}
Giesbrecht and Roche~\cite{GR2011} studied the problem of detecting and computing perfect powers of lacunary polynomials. Their work provides a Monte Carlo algorithm (coRP) for detecting whether $f = h^r$ over $\mathbb{Z}[x]$ and $\mathbb{F}_q[x]$ (for large characteristic), running in polynomial time in the lacunary size. For computing the root $h$, they give two approaches: an output-sensitive deterministic algorithm that depends on the sparsity of $h$, and a faster algorithm based on sparse Newton iteration that relies on a conjecture of Schinzel. Our work differs in several essential aspects: we provide a fully deterministic algorithm for exact-root computation; our complexity bound does not depend on the unknown sparsity of $g$; and we handle the characteristic-dividing case rigorously without conjectural assumptions. The sparsity bound we prove, $\|g\|_0 \le s^{D(2d+2)/e + 1}$, complements the conjectural sparsity structure discussed in their work.

Bisht and Volkovich~\cite{BishtVolkovich2025} further developed structural results for sparse polynomial factorization-related problems, introducing the notions of reverse-monic and pseudo-monic polynomials. Their Lemma~5.3 provides a sparsity bound for exact roots under a reverse-monic condition: if $H$ is $s$-sparse, reverse-monic, has an individual degree of at most $\Delta$, and $H = R^e$, then $\|R\|_0 \le s^{\Delta/e + 1}$. This lemma serves as a crucial technical ingredient in our proof. Our contribution extends their structural bound to arbitrary (not necessarily reverse-monic) exact powers via a scalar weighted substitution and supplies a complete deterministic algorithm with explicit complexity accounting, including the treatment of positive characteristic via true inverse Frobenius roots.

Our result (1) can be viewed as a multi-variate version of
%structural complement to
the Erd\H{o}s--R\'enyi conjecture. Our bound applies to the multivariate setting, where Schinzel's theorem is not directly available and holds over arbitrary fields $\mathbb{F}$ without restriction on the characteristic. A natural question is whether the dependence on the degree $d$ in our bound is necessary. Example~1.1 shows that it is: over $\mathbb{F}_2$, we construct a family $f = g^3$ with $\|f\|_0 = s = (n+1)^2$ and $\|g\|_0 = s^{(\log_2(d-1) + 1)/4}$, demonstrating that the root sparsity can indeed grow as $s^{\Theta(\log d)}$. Thus, for arbitrary fields, no sparsity bound for exact roots can be independent of $d$.

%Our work is structurally complementary to these results: while Schinzel and Zannier give lower bounds on the number of terms of a power in terms of the base, we prove an upper bound on the sparsity of the base $g$ in an exact power $f = g^e$ in terms of the sparsity of $f$. Together, these results provide a more complete picture of the sparsity trade-off between a polynomial and its powers.

\paragraph{Schinzel's conjecture on composite sparse polynomials}
%\textit{Schinzel's conjecture on composite lacunary polynomials.}
Schinzel  generalized the Erd\H{o}s--R\'enyi conjecture to arbitrary polynomial compositions: if $g(h(x))$ has a bounded number of terms for a fixed non-constant polynomial $g$, then $h(x)$ must also have an unbounded number of terms. This conjecture was fully resolved by Zannier~\cite{Zannier2008}, who proved that if $g(h(x))$ has at most $l$ terms, then the number of terms of $h$ is bounded by a function depending only on $l$. 
%Zannier's proof introduces a method based on Puiseux expansions and $S$-unit equations in function fields, which is fundamentally different from Schinzel's earlier approach for powers. 

Our exact-power problem $f = g^e$ can be viewed as the special case of composition where $g(x) = x^e$. However, our focus is algorithmic: we give a deterministic polynomial-time algorithm to compute $g$ given $f$ and $e$, whereas Zannier's work is primarily existential and qualitative, establishing finiteness and parametric descriptions rather than efficient root extraction.

\paragraph{Sparsity of factors of  sparse polynomials}
%
%\textit{Deterministic factorization of sparse polynomials.}
%The most closely related framework to our work is the deterministic factorization algorithm for sparse polynomials with bounded individual degree due to Bhargava, Saraf, and Volkovich~\cite{BSV2020}. 
Bhargava, Saraf, and Volkovich~\cite{BSV2020} established a general factor-sparsity bound: if an $s$-sparse polynomial has individual degree at most $d$, then every factor is sufficiently sparse to permit deterministic reconstruction. 
In particular, they show that if $f$ is an $s$-sparse polynomial in $n$ variables, with individual degrees of its variables bounded by $d$, then the sparsity of each factor of $f$ is bounded by $s^{O(d^2 \log n)}$.
However, this complexity is quasi-polynomial in the sparse parameters for fixed  $d$. 
In order to obtain complexity that is polynomial in the sparse parameters and $d$, Bisht and Volkovich~\cite{BishtVolkovich2025} proposed the following conjecture.

\textbf{Conjecture 1.3}~\cite{BishtVolkovich2025}.
There exists a universal constant $k \in \mathbb{N}$ such that for any $s, d \in \mathbb{N}$, any factor of an $s$-sparse polynomial with individual degree bounded by $d$ has at most $s^{d^k}$ terms.

We solve conjecture in two special cases: \eqref{eq-bd1} exact powers $f = g^e$ and \eqref{eq-bd2} oriented factorizations $f = g h$ with $h$ multilinear. In both cases, we obtain a sparsity bound on $g$ of the form $s^{O(Dd)}$, where $D = \operatorname{tdeg}(f)$ is the total degree of $f$.

\subsection{Paper Organization}

Section 2 gives the preliminary definitions. Section 3 defines the $z_j$-reverse-monic concept, provides a nontrivial example, and relates it to $I$-reverse pseudo-monicity. Section 4 states the transformation lemmas and the exact-power sparsity theorem while identifying the additional lemma needed for a stronger simple-substitution claim. Section 5 presents the multilinear-factor and cofactor bounds for factorizations $f=gh$. Section 6 states the exact-power algorithmic consequence. Section 7 outlines the deterministic factorization framework for the multilinear setting and explains the limitations behind simplified runtimes. Section 8 concludes.

\section{Preliminaries}

Throughout, $\F$ is a field and $\F[x_1,\ldots,x_n]$ is a polynomial ring with $n\geq 1$. 

\begin{definition}[Sparse polynomial and degrees]\label{def:sparse}
A polynomial $f\in \F[x_1,\ldots,x_n]$ is called $s$-sparse if it has at most $s$ monomials with nonzero coefficients. We write $\|f\|_0=|\supp(f)|$ for the number of nonzero terms of $f$. For a nonzero polynomial
\[
  f=\sum_{i=1}^{t} c_i x_1^{d_{i,1}}x_2^{d_{i,2}}\cdots x_n^{d_{i,n}},
\]
where $c_i\in \F^*$ and $t=\|f\|_0$, its total degree is
\[
  \tdeg(f)=\max_{1\leq i\leq t}\left\{\sum_{j=1}^n d_{i,j}\right\},
\]
and its individual degree is
\[
  \indeg(f)=\max_{1\leq i\leq t,\,1\leq j\leq n} d_{i,j}.
\]
We use $D=\tdeg(f)$ and $d=\indeg(f)$ when these parameters are part of the statement under discussion.
\end{definition}

\begin{definition}[$x_j$-reverse-monic]\label{def:xj}
Let $f\in \F[x_1,\ldots,x_n]$. For $j\in[n]$, the polynomial $f$ is $x_j$-reverse-monic if
\[
  f|_{x_j=0}=1.
\]
Equivalently, when $f$ is viewed as a polynomial in $x_j$ with coefficients in the remaining variables, its constant coefficient is $1$.
\end{definition}

\begin{definition}[$I$-reverse pseudo-monic]\label{def:pseudo}
Let $I\subseteq [n]$. A polynomial $f\in \F[x_1,\ldots,x_n]$ is $I$-reverse pseudo-monic if the specialization $f|_{x_i=0\,:\,i\in I}$ is a nonzero field element or a single monomial. This is the convention of Bisht--Volkovich \cite[Definition 2.1]{BishtVolkovich2025}.
\end{definition}

\begin{example}[Example of $x_j$-reverse-monic]\label{ex:xj}
Let
\[
  f(x_1,x_2,x_3)=1+x_2(x_1^2+x_1x_3+x_3^2)+x_2^2x_1x_3.
\]
Then
\[
  f|_{x_2=0}=1,
\]
so $f$ is $x_2$-reverse-monic. This example is not merely univariate in $x_2$: the positive $x_2$-degree part contains the mixed terms $x_1^2x_2$, $x_1x_2x_3$, $x_2x_3^2$, and $x_1x_2^2x_3$.
\end{example}

\begin{remark}[Relationship with $I$-reverse pseudo-monic]\label{rem:relationship}
Every $x_j$-reverse-monic polynomial is $\{j\}$-reverse monic and hence $\{j\}$-reverse pseudo-monic. The pseudo-monic condition is weaker: for instance,
\[
  f(x_1,x_2,x_3)=x_3^4+x_1x_3^2+x_2^2
\]
is $\{1,2\}$-reverse pseudo-monic because $f|_{x_1=x_2=0}=x_3^4$, a single monomial. It is not $\{1,2\}$-reverse monic unless that monomial is normalized to $1$. The source cofactor argument uses this distinction: pseudo-monic factors can be transformed into reverse-monic factors while preserving the relevant sparsities \cite[Lemma 6.16]{BishtVolkovich2025}.
\end{remark}

We use the one source theorem repeatedly. It is the reverse-monic exact-root estimate appearing in the supplied exact-root questions report and in Bisht--Volkovich \cite[Lemma 5.3]{BishtVolkovich2025}.

\begin{theorem}\cite[Lemma 5.3]{BishtVolkovich2025}\label{reversemo}
Let $f\in \mathbb{F}[x_1,\ldots,x_n]$ be an $s$-sparse polynomial of individual degree $\Delta$ which is $x_i$-reverse monic, for some $i\in[n]$. If $f=g^e$ for some polynomial $g\in \mathbb{F}[x_1,\ldots,x_n]$ and $e\in\mathbb{N}$, then $g$ is $s^{\Delta/e+1}$-sparse.
\end{theorem}

The corresponding promised algorithm computes $g$ in $\poly(s^{\Delta/e},n,\Delta)$ field operations \cite[Lemmas 5.3 and 5.8]{BishtVolkovich2025}.

\begin{lemma}[Binomial Truncated Root: Existence and Uniqueness]\label{binomialTrun}
Let $\F$ be a field, and let $e \ge 1$ be an integer invertible in $\F$ (i.e., $\operatorname{char}(\F) \nmid e$). Let $H \in \F[x_1,\ldots,x_n,y]$ with $H(x_1,\dots,x_n,0)=1$, and let $N \ge 0$ be an integer. In the quotient ring $A_N = \F[x_1,\ldots,x_n,y]/(y^{N+1})$, define $R_N = \sum_{i=0}^{N} \binom{1/e}{i} (H-1)^i \in A_N$, where the binomial coefficients $\binom{1/e}{i} = \frac{(1/e)(1/e-1)\cdots(1/e-i+1)}{i!}$ are well-defined in $\F$ since $e$ is invertible. Then $R_N$ is the \emph{unique} element in $A_N$ satisfying $R_N(x_1,\dots,x_n,0)=1$ and $R_N^e = H$ in $A_N$.
\end{lemma}

\begin{proof}

\textbf{Existence.} Let $T = H-1$. Since $H(x_1,\dots,x_n,0)=1$, we have $T \equiv 0 \pmod{y}$. Thus
\[
R_N = \sum_{i=0}^{N} \binom{1/e}{i} T^i \in A_N.
\]
We verify directly that $R_N^e = H$ in $A_N$. In the formal power series ring $\F[x_1,\ldots,x_n][[y]]$, the binomial series $(1+T)^{1/e} = \sum_{i=0}^{\infty} \binom{1/e}{i} T^i$ satisfies
\[
\left(\sum_{i=0}^{\infty} \binom{1/e}{i} T^i\right)^e = 1+T = H.
\]
Since $T$ is divisible by $y$, every term $T^i$ has $y$-degree at least $i$. Thus, modulo $y^{N+1}$, all terms with $i > N$ contribute nothing. Therefore,
\[
R_N^e
= \left(\sum_{i=0}^{N} \binom{1/e}{i} T^i\right)^e
\equiv \left(\sum_{i=0}^{\infty} \binom{1/e}{i} T^i\right)^e
= H
\pmod{y^{N+1}},
\]
i.e., $R_N^e = H$ in $A_N$. Moreover, $R_N(x_1,\dots,x_n,0) = \binom{1/e}{0} \cdot 1 = 1$. This proves existence.

\textbf{Uniqueness.} Suppose $U,V \in A_N$ both satisfy
\[
U(x_1,\dots,x_n,0)=V(x_1,\dots,x_n,0)=1, \qquad U^e = V^e.
\]
We prove that $U=V$.

Assume for contradiction that $U \ne V$. Since every element of $A_N$ has a unique representative of $y$-degree at most $N$, there exists a smallest integer $t$ with $0 \le t \le N$ such that the coefficients of $y^t$ in $U$ and $V$ differ. Since $U(x_1,\dots,x_n,0)=V(x_1,\dots,x_n,0)=1$, we have $t \ge 1$. Then $U - V = y^t W$, where $W \in A_N$ and $W$ is not divisible by $y$ (i.e., $W(x_1,\dots,x_n,0) \ne 0$).

On the other hand, since $U \equiv V \equiv 1 \pmod{y}$, we have
\[
U^{e-1} + U^{e-2}V + \cdots + V^{e-1} \equiv \underbrace{1 + \cdots + 1}_{e \text{ terms}} = e \pmod{y}.
\]
Because $e$ is invertible in $\F$, $e \not\equiv 0 \pmod{y}$, so the above sum is a unit in $A_N$.

Now compute the difference:
\[
U^e - V^e = (U-V)(U^{e-1} + U^{e-2}V + \cdots + V^{e-1}) = y^t W \cdot S,
\]
where $S = U^{e-1} + U^{e-2}V + \cdots + V^{e-1}$. Since $S(x_1,\dots,x_n,0)=e \ne 0$ and $W(x_1,\dots,x_n,0) \ne 0$, the product $W \cdot S$ is not divisible by $y$. Hence, $U^e - V^e$ has a nonzero coefficient at $y^t$ (with $t \le N$). In particular, $U^e - V^e \ne 0$ in $A_N$. This contradicts the assumption $U^e = V^e$. Therefore, $U=V$. Uniqueness is proved.
\end{proof}

\begin{corollary}[Complete Recovery in the Exact Power Case]\label{binomialExactPower}
Let $\F$ be a field, and let $e \ge 1$ be an integer invertible in $\F$ (i.e., $\operatorname{char}(\F) \nmid e$). Assume $H \in \F[x_1,\ldots,x_n,y]$ satisfies $H(x_1,\dots,x_n,0)=1$, and suppose there exists $R \in \F[x_1,\ldots,x_n,y]$ such that $H = R^e$ and $\tdeg_y(H) \le N$. Let $M = \lfloor N/e \rfloor$. Then
\[
R \equiv \sum_{i=0}^{M} \binom{1/e}{i} (H-1)^i \pmod{y^{M+1}}.
\]
In particular, since $\tdeg_y(R) \le M$, this congruence determines $R$ completely in $\F[x_1,\ldots,x_n,y]$.
\end{corollary}

\begin{proof}
Since $H = R^e$, we have $\tdeg_y(H) = e \cdot \tdeg_y(R)$, hence $\tdeg_y(R) = \tdeg_y(H)/e \le N/e$, so $\tdeg_y(R) \le \lfloor N/e \rfloor = M$.
 In the quotient ring $A_M = \F[x_1,\ldots,x_n,y]/(y^{M+1})$, define
\[
\widetilde{R} = \sum_{i=0}^{M} \binom{1/e}{i} (H-1)^i.
\]
The Lemma \ref{binomialTrun} guarantees that $\widetilde{R}$ is the unique element in $A_M$ satisfying $\widetilde{R}(x_1,\dots,x_n,0)=1$ and $\widetilde{R}^e = H$ in $A_M$. The polynomial $R$ itself satisfies $R(x_1,\dots,x_n,0)=1$ (because $R(x_1,\dots,x_n,0)^e = H(x_1,\dots,x_n,0)=1$, and $e$ is invertible in $\F$) and $R^e = H$ in $A_M$. Hence $\widetilde{R} = R$ in $A_M$, i.e.,
\[
R \equiv \sum_{i=0}^{M} \binom{1/e}{i} (H-1)^i \pmod{y^{M+1}}.
\]
Since $\tdeg_y(R) \le M$, this congruence determines all coefficients of $R$ completely.
\end{proof}

\section{\texorpdfstring{$I$-Reverse Pseudo-Monic Transformation and Sparsity Bound}{I-Reverse Pseudo-Monic Transformation and Sparsity Bound}}

A useful way to approach arbitrary exact powers is to separate a lowest term and transform the polynomial toward reverse or pseudo-monic form. For the exact-power theorem below we use the scalar weighted substitution directly. After a separating weight is chosen, division by the minimum weight of the support produces a $\{y\}$-reverse pseudo-monic exact power. 
This constructs the required pseudo-monic input for the exact-root lemma, and avoids the earlier approach of forcibly converting the polynomial into reverse-monic form before applying the estimate.

We first isolate a support point by a weight vector whose coordinates are
bounded only in terms of \(d\).

\begin{lemma}[Separating vector]\label{lem:separator}
Let \(d\ge0\), \(n\ge1\), and let
\[
  A\subseteq\{0,\ldots,d\}^n
\]
be finite and nonempty. Then there are \(\mathbf{a}_0\in A\) and
\[
  \mathbf{k}\in\{1,\ldots,2d+2\}^n
\]
such that
\[
  \langle \mathbf{a}_0,\mathbf{k}\rangle<\langle \mathbf{a},\mathbf{k}\rangle
\]
for all \(\mathbf{a}\in A\setminus\{\mathbf{a}_0\}\).
\end{lemma}

\begin{proof}
If \(d=0\), then \(A=\{(0,\ldots,0)\}\). Any \(\mathbf{a}_0\in A\) and
\(\mathbf{k}=(1,\ldots,1)\) work.

Assume \(d\ge1\). Define
\[
  \phi(0)=0,\qquad
  \phi(t)=\sum_{m=0}^{t-1}\bigl(2(d-m)+1\bigr)\quad (1\le t\le d).
\]
Then
\[
  \phi(r+1)-\phi(r)=2(d-r)+1\qquad (0\le r\le d-1).
\]
For \(t,v\in\{0,\ldots,d\}\) with \(v\ne t\), this implies
\[
  \phi(v)-\phi(t)<2(d-t+1)(v-t).
\]
Indeed, if \(v>t\), telescope from \(t\) to \(v\); each summand is strictly
less than \(2(d-t+1)\). If \(v<t\), telescope from \(v\) to \(t\); each
summand is strictly greater than \(2(d-t+1)\), and multiply the resulting
inequality by \(-1\).

For \(\mathbf{a}=(a_1,\ldots,a_n)\), put
\[
  \Phi(\mathbf{a})=\sum_{i=1}^n\phi(a_i).
\]
Choose \(\mathbf{a}_0=(a_{01},\ldots,a_{0n})\in A\) minimizing \(\Phi\). Define
\[
  k_i=2(d-a_{0i}+1)\qquad (1\le i\le n).
\]
Then \(2\le k_i\le2d+2\), so \(\mathbf{k}\in\{1,\ldots,2d+2\}^n\). If
\(\mathbf{a}\in A\setminus\{\mathbf{a}_0\}\), let \(S=\{i:a_i\ne a_{0i}\}\). Summing the
one-coordinate inequality over the nonempty set \(S\) gives
\[
  \Phi(\mathbf{a})-\Phi(\mathbf{a}_0)
  <\sum_{i\in S}k_i(a_i-a_{0i})
  =\langle \mathbf{a},\mathbf{k}\rangle-\langle \mathbf{a}_0,\mathbf{k}\rangle.
\]
Since \(\mathbf{a}_0\) minimizes \(\Phi\), the left-hand side is nonnegative. Hence
\[
  \langle \mathbf{a}_0,\mathbf{k}\rangle<\langle \mathbf{a},\mathbf{k}\rangle,
\]
as required.
\end{proof}

\begin{theorem}[Scalar pseudo-monic exact-power transform]\label{thm:unified_sep}
Let $\mathbb F$ be a field, let $f,g\in \mathbb F[x_1,\ldots,x_n]\setminus \mathbb F$, let $e\in \mathbb N$ with $e\ge 1$, and suppose $f=g^e$. Assume
\[
1\le \operatorname{indeg}(f)\le d,\qquad \tdeg(f)\le D.
\]
Then there are $\mathbf{a}_0\in \operatorname{supp}(f)$ and a vector
\[
\mathbf{k}=(k_1,\ldots,k_n)\in\{1,\ldots,2d+2\}^n
\]
such that $\mathbf{a}_0$ is the unique lowest element of $\operatorname{supp}(f)$ for the weight $\langle \mathbf{a},\mathbf{k}\rangle$. If $c_0\mathbf{x}^{\mathbf{a}_0}$ is this lowest monomial of $f$, then the lowest $\mathbf{k}$-weight part of $g$ is a single monomial $\gamma \mathbf{x}^{\mathbf{b}_0}$ with $\gamma\in \mathbb F^*$ and
\[
\mathbf{a}_0=e\mathbf{b}_0,\qquad c_0=\gamma^e,\qquad \langle \mathbf{a}_0,\mathbf{k}\rangle=e\langle \mathbf{b}_0,\mathbf{k}\rangle.
\]
Writing
\[
k_f=\langle \mathbf{a}_0,\mathbf{k}\rangle,\qquad k_g=\langle \mathbf{b}_0,\mathbf{k}\rangle,
\]
define
\[
F^\sharp=y^{-k_f}f(x_1y^{k_1},\ldots,x_ny^{k_n}),\qquad
G^\sharp=y^{-k_g}g(x_1y^{k_1},\ldots,x_ny^{k_n}).
\]
Then $F^\sharp,G^\sharp\in \mathbb F[x_1,\ldots,x_n,y]$,
\[
F^\sharp=(G^\sharp)^e,\qquad F^\sharp|_{y=0}=c_0\mathbf{x}^{\mathbf{a}_0},
\]
and
\[
\|F^\sharp\|_0=\|f\|_0,\qquad \|G^\sharp\|_0=\|g\|_0.
\]
In particular, $F^\sharp$ is $\{y\}$-reverse pseudo-monic. Moreover, we have the degree bound
\[
\operatorname{indeg}(F^\sharp)\le D(2d+2).
\]

\end{theorem}

\begin{proof}
The separating vector and the unique lowest monomial $\mathbf{a}_0$ are supplied by Lemma \ref{lem:separator}. The lowest-root-monomial argument applied to $f=g^e$ gives the displayed monomial $\gamma \mathbf{x}^{\mathbf{b}_0}$ of $g$ and the identities
\[
\mathbf{a}_0=e\mathbf{b}_0,\qquad c_0=\gamma^e,\qquad k_f=e k_g.
\]

For a monomial $c_{\mathbf{a}}\mathbf{x}^{\mathbf{a}}$ of $f$, its image in $F^\sharp$ is
\[
c_{\mathbf{a}}\mathbf{x}^{\mathbf{a}} y^{\langle \mathbf{a},\mathbf{k}\rangle-k_f}.
\]
The definition of $k_f$ as the minimum $\mathbf{k}$-weight on $\operatorname{supp}(f)$ makes every exponent $\langle \mathbf{a},\mathbf{k}\rangle-k_f$ nonnegative, so $F^\sharp$ is an ordinary polynomial. The same argument, using the minimum weight $k_g$ on $\operatorname{supp}(g)$, shows that $G^\sharp\in \mathbb F[x_1,\ldots,x_n,y]$.

The scalar substitution $p\mapsto p(x_1y^{k_1},\ldots,x_ny^{k_n})$ is a ring homomorphism. Since $f=g^e$ and $k_f=e k_g$,
\[
F^\sharp
=y^{-k_f}f(x_1y^{k_1},\ldots,x_ny^{k_n})
=y^{-e k_g}g(x_1y^{k_1},\ldots,x_ny^{k_n})^e
=(G^\sharp)^e.
\]
The specialization at $y=0$ keeps exactly the terms of minimum $\mathbf{k}$-weight. The minimum in $f$ is unique, so
\[
F^\sharp|_{y=0}=c_0\mathbf{x}^{\mathbf{a}_0},
\]
a single nonzero monomial. Thus $F^\sharp$ is $\{y\}$-reverse pseudo-monic.

Finally, the exponent map
\[
\mathbf{a}\longmapsto (\mathbf{a},\langle \mathbf{a},\mathbf{k}\rangle-k_f)
\]
from monomials of $f$ to monomials of $F^\sharp$ is injective because its first $n$ coordinates recover $\mathbf{a}$. No nonzero coefficient is killed, so $\|F^\sharp\|_0=\|f\|_0$. The identical argument for
\[
\mathbf{b}\longmapsto (\mathbf{b},\langle \mathbf{b},\mathbf{k}\rangle-k_g)
\]
gives $\|G^\sharp\|_0=\|g\|_0$.

It remains to prove the degree bound. Let
\[
H(x_1,\ldots,x_n,y)=f(x_1y^{k_1},\ldots,x_ny^{k_n}).
\]
For any monomial $\mathbf{x}^{\mathbf{a}}=x_1^{a_1}\cdots x_n^{a_n}$ of $f$, its image under the substitution $x_j\mapsto x_jy^{k_j}$ is
\[
x_1^{a_1}\cdots x_n^{a_n}\, y^{a_1k_1+\cdots+a_nk_n}.
\]
Since $k_j\le 2d+2$ for every $j$,
\[
a_1k_1+\cdots+a_nk_n
\le (a_1+\cdots+a_n)(2d+2)
\le D(2d+2),
\]
because the total degree of every monomial in $f$ is at most $D=\tdeg(f)$. Hence
\[
\tdeg_y(H)\le D(2d+2).
\]
Since $F^\sharp=y^{-k_f}H$ and $k_f\ge 0$, we have
\[
\tdeg_y(F^\sharp)\le \tdeg_y(H)\le D(2d+2).
\]

Moreover, for each variable $x_i$ ($1\le i\le n$), the substitution $x_j\mapsto x_jy^{k_j}$ does not change the exponent of $x_i$ in any monomial, and the subsequent division by $y^{-k_f}$ does not affect the $x_i$-exponents either. Hence
\[
\tdeg_{x_i}(F^\sharp)=\tdeg_{x_i}(f)\qquad\text{for every }1\le i\le n.
\]
In other words, the individual degrees in the original variables $x_1,\ldots,x_n$ remain unchanged under the construction of $F^\sharp$. Thus $\operatorname{indeg}(F^\sharp)\le D(2d+2)$, as claimed.
\end{proof}

The following lemma is a local version of the pseudo-monic exact-root estimate. While Bisht--Volkovich \cite[Theorem 6.19]{BishtVolkovich2025} gives an asymptotic exponent for general $I$-reverse pseudo-monic inputs, we need the exact exponent stated below. It is obtained by combining the reverse-monic binomial expansion argument \cite[Lemma 5.3]{BishtVolkovich2025} with the pseudo-monic normalization technique \cite[Lemma 6.16 and Remark 6.18]{BishtVolkovich2025}.

\begin{lemma}[Local pseudo-monic exact-root bound]\label{lem:pseudo-exact-root}
Let $H,R\in \F[x_1,\ldots,x_n,y]$ satisfy $H=R^e$ with $e\geq 1$. Suppose that $H$ is $S$-sparse, that
\[
  H(x_1,\ldots,x_n,0)
\]
is a nonzero field element or a single nonzero monomial, and that $\tdeg_y(H)\leq \Delta$. Then
\[
  \|R\|_0\leq S^{\Delta/e+1}.
\]
\end{lemma}

\begin{proof}

First consider the reverse-monic case $H(x_1,\dots,x_n,0)=1$. Substituting $y=0$ in $H=R^e$ shows that $R(x_1,\dots,x_n,0)^e=1$, so $R(x_1,\dots,x_n,0)$ is a nonzero constant. Multiplying $R$ by this constant's inverse does not change support; after this scalar normalization we may count support with $R(x,0)=1$.

If $\operatorname{char}(\F)$ is zero or does not divide $e$, the binomial expansion
\[
  R=H^{1/e}=(1+(H-1))^{1/e}
\]
is valid modulo powers of the ideal $(y)$. Every monomial of $H-1$ has positive $y$-degree, while $\tdeg_y(R)=\tdeg_y(H)/e\leq \Delta/e$. Hence terms of $(H-1)^j$ with $j>\Delta/e$ cannot contribute to $R$, and
\[
  R\equiv
  \sum_{j=0}^{\lfloor \Delta/e\rfloor}\binom{1/e}{j}(H-1)^j
  \pmod{(y)^{\lfloor \Delta/e\rfloor+1}}.
\]
Passing to this quotient cannot increase support. If $S=1$ the claim is immediate; otherwise $\|H-1\|_0\leq S$ and the finite sum has support at most
\[
  \sum_{j=0}^{\lfloor \Delta/e\rfloor}S^j\leq S^{\Delta/e+1}.
\]

If $\operatorname{char}(\F)=p>0$ and $p$ divides $e$, write $e=p^kq$ with $p\nmid q$. The $p^k$-th power map sends monomial exponents to $p^k$ times those exponents and does not change support. Taking the unique $p^k$-th root of all exponent vectors occurring in $H$ gives a polynomial $H_0$ with $\|H_0\|_0=\|H\|_0$, $\tdeg_y(H_0)=\tdeg_y(H)/p^k$, and $H_0=R_0^q$ for a polynomial $R_0$ with $\|R_0\|_0=\|R\|_0$. The preceding case applied to $H_0=R_0^q$ gives
\[
  \|R\|_0=\|R_0\|_0\leq S^{(\Delta/p^k)/q+1}=S^{\Delta/e+1}.
\]

It remains to reduce the pseudo-monic case to the reverse-monic case. Put
\[
  \alpha=H(x_1,\dots,x_n,0).
\]
By hypothesis $\alpha$ is a nonzero scalar or a single nonzero monomial. Since $H=R^e$, the specialization $\beta=R(x_1,\dots,x_n,0)$ satisfies $\beta^e=\alpha$. A polynomial whose $e$-th power is a nonzero scalar multiple of a monomial is itself a nonzero scalar multiple of a monomial, so $\beta$ has the same pseudo-monomial form.

Define
\[
  \Psi(P)=P(x_1,\ldots,x_n,y\alpha),\qquad
  H^*=\alpha^{-1}\Psi(H),\qquad R^*=\beta^{-1}\Psi(R).
\]
For a term of $H$ with $y$-degree $j\geq 1$, the substitution contributes the monomial factor $\alpha^j$, and division by $\alpha$ leaves $\alpha^{j-1}$; the unique $y$-degree zero term becomes $1$. Thus $H^*$ is a polynomial. The same argument with $\beta$ shows that $R^*$ is a polynomial. Moreover,
\[
  H^*=\alpha^{-1}\Psi(R^e)=\beta^{-e}\Psi(R)^e=(R^*)^e
\]
and $H^*(x_1,\dots,x_n,0)=1$.

The transformation preserves the exponent of $y$ in every monomial, so $\tdeg_y(H^*)=\tdeg_y(H)\leq\Delta$. It also preserves support. Indeed, for fixed $y$-degree $j$, the exponent of the $x$-part is translated by a fixed vector, and the $y$-degree itself is retained; hence distinct monomials remain distinct. Therefore
\[
  \|H^*\|_0=\|H\|_0\leq S,\qquad \|R^*\|_0=\|R\|_0.
\]
Applying the reverse-monic case to $H^*=(R^*)^e$ yields
\[
  \|R\|_0=\|R^*\|_0\leq S^{\Delta/e+1}.
\]

\end{proof}

\begin{theorem}[\texorpdfstring{Sparsity bound for root in $f=g^e$}{Sparsity bound for transformed/root in f=g^e}]\label{thm:exact}
Let $\F$ be a field, let $s,d\in\mathbb N$, let $e \ge 1$ be an integer, and let
$f,g\in \F[x_1,\ldots,x_n]$ satisfy
\[
  f=g^e,\qquad \|f\|_0\leq s,\qquad \indeg(f)\leq d,\qquad \tdeg(f)\leq D
\]
Then
\[
  \|g\|_0\leq s^{D(2d+2)/e+1}.
\]
\end{theorem}

\begin{proof}
If $f\in \F$, the bound is trivial. Assume $f\notin \F$, then $f$ is nonzero and nonconstant. Hence
$\supp(f)$ is finite and nonempty, and the hypothesis $\indeg(f)\leq d$
forces $d\geq 1$. Apply Theorem~\ref{thm:unified_sep} and use its notation. It gives ordinary polynomials $F^\sharp,G^\sharp\in \F[x_1,\ldots,x_n,y]$ such that
\[
  F^\sharp=(G^\sharp)^e,\qquad
  F^\sharp|_{y=0}=c_0x^{a_0},\qquad
  \|F^\sharp\|_0=\|f\|_0\leq s,\qquad
  \|G^\sharp\|_0=\|g\|_0.
\]
Thus $F^\sharp$ is $\{y\}$-reverse pseudo-monic and
\[
  \tdeg_y(F^\sharp)\leq D(2d+2).
\]
Lemma~\ref{lem:pseudo-exact-root}, applied with
\[
  H=F^\sharp,\qquad R=G^\sharp,\qquad S=s,\qquad \Delta=D(2d+2),
\]
first gives the sharper estimate
\[
  \|G^\sharp\|_0\leq s^{D(2d+2)/e+1}.
\]
This implies
\[
  \|g\|_0=\|G^\sharp\|_0\leq s^{D(2d+2)/e+1},
\]
which is the stated bound.
\end{proof}

\section{\texorpdfstring{Sparsity Bound for $f=gh$ with $h$ Multilinear}{Sparsity Bound for f=gh with h Multilinear}}

We now turn to multilinear cofactors. The orientation matters: the named factor $h$ is the multilinear divisor in $f=gh$.

\begin{corollary}[$I$-reverse pseudo-monic cofactor result]\label{cor:pseudo}
Let $f=gh$ and suppose that $h$ is $I$-reverse pseudo-monic. Suppose also that the individual degrees of the variables of $g$ in $x_I$ are at most $d$, and that $f$ and $h$ are both bounded by the same sparsity parameter $s$ in the setup of the source theorem. Then
\[
  \|g\|_0\leq s^{d|I|+2}.
\]
\end{corollary}

\begin{proof}
This is Corollary 6.17 of Bisht--Volkovich \cite{BishtVolkovich2025}. A different assertion, namely that pseudo-monicity of a product automatically transfers to both factors, is not the cofactor bound used here. For exact powers, the source separately observes that $f=g^e$ preserves $I$-reverse pseudo-monicity between $f$ and $g$; for cofactors, Corollary 6.17 requires pseudo-monicity of the factor $h$ and bounds $g$.
\end{proof}

\begin{theorem}[\texorpdfstring{Sparsity bound for $g$ with multilinear $h$}{Sparsity bound for g with multilinear h}]\label{thm:multi}
Let $\F$ be a field, let $f,g\in \F[x_1,\ldots,x_n]\setminus \F$, and let $h\in \F[x_1,\ldots,x_n]$ be such that $f,g,h$ form a nonzero oriented factorization
\[
  f=gh.
\]
Suppose that the named factor $h$ is multilinear, that $\|f\|_0\leq s$, and that $\indeg(f)\leq d$. Write $D=\tdeg(f)$. Then
\[
  \|h\|_0\leq s,\qquad
  \|g\|_0\leq s^{(2d+2)D+2}.
\]
\end{theorem}

\begin{proof}
Because $h$ is multilinear and divides $f$, Lemma 2.2 of Bisht--Volkovich gives
\[
  \|h\|_0\leq \|f\|_0\leq s
\]
\cite[Lemma 2.2]{BishtVolkovich2025}. For each variable $x_i$, degree in $x_i$ is additive in the nonzero product $f=gh$ after viewing the polynomial ring as a univariate polynomial ring over the remaining variables. Hence $\tdeg_{x_i}(g)\leq\tdeg_{x_i}(f)\leq d$, and also $\tdeg(g)\leq\tdeg(f)=D$.

Choose a vector $\mathbf{k}=(k_1,\ldots,k_n)\in\{1,\ldots,2d+2\}^n$ so that the weight $\langle \mathbf{a},\mathbf{k}\rangle$ has a unique minimum on $\operatorname{supp}(f)$; this is the bounded separating-vector lemma applied to $\operatorname{supp}(f)\subseteq\{0,\ldots,d\}^n$. Let $k_f,k_g,k_h$ be the minimum $\mathbf{k}$-weights of $f,g,h$, respectively, and define
\[
  F^\sharp=y^{-k_f}f(x_1y^{k_1},\ldots,x_ny^{k_n}),\quad
  G^\sharp=y^{-k_g}g(x_1y^{k_1},\ldots,x_ny^{k_n}),\quad
  H^\sharp=y^{-k_h}h(x_1y^{k_1},\ldots,x_ny^{k_n}).
\]
All three transformed polynomials lie in $\mathbb{F}[x_1,\ldots,x_n,y]$. Minimum weights in the nonzero product give $k_f=k_g+k_h$, so $F^\sharp=G^\sharp H^\sharp$. Since the initial $\mathbf{k}$-form of $f$ is one monomial and the polynomial ring is a domain, the initial $\mathbf{k}$-form of $h$ is also one monomial; equivalently, $H^\sharp|_{y=0}$ is a nonzero monomial, so $H^\sharp$ is $\{y\}$-reverse pseudo-monic. The scalar substitution preserves supports, hence
\[
  \|F^\sharp\|_0=\|f\|_0\leq s,\qquad
  \|H^\sharp\|_0=\|h\|_0\leq s,\qquad
  \|G^\sharp\|_0=\|g\|_0.
\]
For any monomial $\mathbf{x}^{\mathbf{b}}$ of $g$, its $y$-exponent in $G^\sharp$ is at most $\langle \mathbf{b},\mathbf{k}\rangle\leq(2d+2)\operatorname{tdeg}(\mathbf{x}^{\mathbf{b}})\leq(2d+2)D$. Thus $\operatorname{tdeg}_y(G^\sharp)\leq(2d+2)D$. Applying Corollary~\ref{cor:pseudo} with $I=\{y\}$ and $d$ there replaced by $(2d+2)D$ gives
\[
  \|G^\sharp\|_0\leq s^{(2d+2)D+2}.
\]
Support preservation gives the same bound for $g$.
\end{proof}

\section{\texorpdfstring{Deterministic Algorithm for Exact Powers $f=g^e$}{Deterministic Algorithm for Exact Powers f equals g to the e}}

\subsection{The Exact-Root Algorithm}

We consider the following deterministic exact \(e\)-th power problem in the polynomial ring \(\F[x_1,\ldots,x_n]\).

The input consists of a nonconstant polynomial
\[
  f\in \F[x_1,\ldots,x_n]
\]
and an integer \(e\ge 2\). The output is required to distinguish the following two cases.

First, if there exists a polynomial \(g\in \F[x_1,\ldots,x_n]\) such that
\[
  f=g^e,
\]
then the algorithm outputs such a base \(g\).

Second, if no such \(g\in \F[x_1,\ldots,x_n]\) exists, then the algorithm outputs ``not an \(e\)-th power''.

\paragraph{The reconstruction algorithm.}
The reconstruction uses the scalar pseudo-monic transformation of Theorem~\ref{thm:unified_sep}, and then applies the \(y\)-adic binomial-root idea of Bisht--Volkovich after normalizing the pseudo-monic initial term.

Choose the separating vector $\mathbf{k}=(k_1,\ldots,k_n)$ and the unique lowest term
$c\mathbf{x}^{\mathbf{a}}$ of $f$ as in  Theorem~\ref{thm:unified_sep}. Put
\[
  k_f=\langle \mathbf{a},\mathbf{k}\rangle,\qquad
  F^\sharp=y^{-k_f}f(x_1y^{k_1},\ldots,x_ny^{k_n}).
\]
Then
\[
  F^\sharp\in \F[x_1,\ldots,x_n,y],\qquad
  F^\sharp=(G^\sharp)^E,\qquad
  F^\sharp(\mathbf{x},0)=c\mathbf{x}^{\mathbf{a}}
\]
for
\[
  G^\sharp=y^{-k_g}g(x_1y^{k_1},\ldots,x_ny^{k_n}).
\]
Write
\[
  \alpha=c\mathbf{x}^{\mathbf{a}},\qquad F^\sharp=\alpha+yF_1
\]
with $F_1\in \F[x_1,\ldots,x_n,y]$.

If the scalar $c$ has no $e$-th root in
$\F$, or if some coordinate of $a$ is not divisible by $e$, then there is no
base in $\F[x_1,\ldots,x_n]$ compatible with this lowest term. In that case the
algorithm returns ``not an \(e\)-th power'' for the exponent $e$, unless the field model permits
an explicitly represented extension field as an intermediate. In the
extension-field variant, the final output is accepted only if the specialized
candidate has all coefficients in $\F$ and passes the exact equality check
$g^e=f$ in $\F[x_1,\ldots,x_n]$.

Suppose first that an $e$-th root of $\alpha$ is available in $\F[x]$. Thus
choose $\beta\in \F^*$ and $\mathbf{b}\in\mathbb N^n$ such that
\[
  \beta^e=c,\qquad e\mathbf{b}=\mathbf{a},
\]
and set $B=\beta \mathbf{x}^{\mathbf{b}}$. Conceptually,
\[
  F^\sharp=B^e\left(1+\frac{yF_1}{\alpha}\right),
\]
so one wants the normalized unit root
\[
  \left(1+\frac{yF_1}{\alpha}\right)^{1/e}.
\]
The quotient by the monomial $\alpha$ may be Laurent in the $x$-variables, so
the polynomial implementation uses the pseudo-monic normalization
\[
  H=\alpha^{-1}F^\sharp(x_1,\ldots,x_n,y\alpha).
\]
For a monomial of $F^\sharp$ with $y$-degree $j\geq 1$, substitution
$y\mapsto y\alpha$ contributes the factor $\alpha^j$, and division by
$\alpha$ leaves $\alpha^{j-1}$. Hence $H$ is an ordinary polynomial in
$\F[x_1,\ldots,x_n,y]$, and
\[
  H(x_1,\dots,x_n,0)=1.
\]
If $G^\sharp$ is the desired transformed root, then
\[
  R=\beta^{-1}\mathbf{x}^{-\mathbf{b}}G^\sharp(x_1,\ldots,x_n,y\alpha)
\]
is also an ordinary polynomial and satisfies $H=R^e$ and $R(x_1,\dots,x_n,0)=1$. Thus $H$
is the reverse-monic input to which the binomial expansion applies after
removing any characteristic-dividing part of $e$.

\begin{algorithm}
\caption{Base reconstruction from the scalar pseudo-monic form}
\label{alg:phase2-base-reconstruction}
\begin{algorithmic}[1]
\Require A nonconstant $f\in \F[x_1,\ldots,x_n]$ and an exponent $e\geq 2$.
\Ensure A polynomial $g\in \F[x_1,\ldots,x_n]$ with $g^e=f$, or output ``not an $e$-th power''.
\State Choose $\mathbf{k}$ and the unique lowest term $c\mathbf{x}^{\mathbf{a}}$ of $f$ as in Theorem~\ref{thm:unified_sep}.
\State Set $k_f=\langle \mathbf{a},\mathbf{k}\rangle$ and $F^\sharp=y^{-k_f}f(x_1y^{k_1},\ldots,x_ny^{k_n})$.
\State Write $\alpha=c\mathbf{x}^{\mathbf{a}}$ and $F^\sharp=\alpha+yF_1$.
\If{$c$ has no $e$-th root in $\F$ or some coordinate of $\mathbf{a}$ is not divisible by $e$}
\State \Return ``not an $e$-th power''.
\EndIf
\State Choose $\beta\in \F^*$ and $\mathbf{b}\in\mathbb N^n$ with $\beta^e=c$ and $e\mathbf{b}=\mathbf{a}$; set $B=\beta \mathbf{x}^{\mathbf{b}}$.
\State Form the reverse-monic polynomial $H=\alpha^{-1}F^\sharp(x_1,\ldots,x_n,y\alpha)$.
\State Set $p=\operatorname{char}(\F)$ and $e'=e$.
\If{$p>0$ and $p\mid e'$}
\State Write $e'=p^r q$ with $p\nmid q$.
\State Take the true $p^r$-Frobenius root of $H$: for each term $A_{\mathbf{u}} \mathbf{x}^{\mathbf{u}} y^{u_y}$, require every coordinate of $(\mathbf{u},u_y)$ to be divisible by $p^r$ and require a coefficient $A_{\mathbf{u}}^{1/p^r}\in \F$ with $(A_{\mathbf{u}}^{1/p^r})^{p^r}=A_{\mathbf{u}}$.
\If{this Frobenius-root operation is not defined over $\F$}
\State \Return ``not an $e$-th power''.
\EndIf
\State Replace $H$ by the resulting polynomial and set $e=q$.
\EndIf
\State Let $\Delta_y=\deg_y(F^\sharp)$ and $N=\lfloor \Delta_y/e\rfloor$.
\State Compute
\[
  R_0=\sum_{i=0}^{N}\binom{1/e}{i}(H-1)^i \pmod {y^{N+1}}.
\]
\State Do not apply a $p^r$-power pullback to $R_0$ after this step.
\State Undo the pseudo-monic normalization by forming the Laurent expression
\[
  G^\sharp_{\mathrm{cand}}=B\,R_0(x_1,\ldots,x_n,y/\alpha).
\]
\If{$G^\sharp_{\mathrm{cand}}\notin \F[x_1,\ldots,x_n,y]$}
\State \Return ``not an $e$-th power''.
\EndIf
\State Set $g_{\mathrm{cand}}=G^\sharp_{\mathrm{cand}}(x_1,\ldots,x_n,1)$.
\If{$g_{\mathrm{cand}}\in \F[x_1,\ldots,x_n]$ and $g_{\mathrm{cand}}^e=f$}
\State \Return $g_{\mathrm{cand}}$.
\Else
\State \Return ``not an $e$-th power''.
\EndIf
\end{algorithmic}
\end{algorithm}

When $\operatorname{char}(\F)=p>0$ divides the active exponent, the binomial
coefficients with $1/e$ are not defined. We write $e=p^rq$ with $p\nmid q$
and use a true inverse Frobenius-root convention. The reduction of $H$ divides
all exponent vectors by $p^r$ and takes $p^r$-th roots of coefficients in $F$:
for
\[
  P=\sum_u A_u x^u y^{u_y}
\]
the reduced polynomial exists over $\F$ only when every coordinate of $(u,u_y)$
is divisible by $p^r$ and each $A_u$ has a $p^r$-th root in $\F$. Since a field
has no nonzero nilpotents, the Frobenius map on $\F$ is injective, so such a
coefficient root is unique when it exists. If this condition fails over $\F$,
the algorithm cannot output a base over $\F$ from this branch. Extension-field
computations, if allowed by the field model, remain only candidate generation,
followed by coefficient membership and exact equality checks over $\F$.

With this convention there is no post-binomial $p^r$-power pullback. Before
the reduction the normalized polynomial has the form
\[
  H=R^e=R^{p^rq}
\]
with $R(x_1,\dots,x_n,0)=1$. In characteristic $p$,
\[
  R^{p^rq}=(R^q)^{p^r}.
\]
Taking the true inverse Frobenius root of $H$ therefore gives $H_0=R^q$. The
active exponent after the reduction is $q$, which is invertible in $\F$, and
the binomial expansion applied to $H_0$ computes the unique normalized
$q$-th root with $y=0$ specialization $1$, namely $R$ itself. Thus $R_0$ is
already the normalized root of the original $H$, and the next step is to undo
the pseudo-monic normalization by forming $B\,R_0(x_1,\dots,x_n,y/\alpha)$. Applying
$R_0(x_1^{p^r},\ldots,x_n^{p^r},y^{p^r})$ at this point would change the root
and is not part of this convention.

This differs from the printed Algorithm 2 reversal in Steps 10--13 because
those steps belong to the source's variable-renaming convention. The present
construction uses the coefficient-and-exponent inverse Frobenius root instead;
once that stronger reduction is used, the source pullback has already been
accounted for algebraically by the equality $H_0=R^q$ and must not be repeated.
For example, if $\operatorname{char}(\F)=p$, $e=p$, and
\[
  H=1+y^p=(1+y)^p,
\]
then the true inverse Frobenius root is $H_0=1+y$ and the active exponent is
$q=1$. The binomial step returns $R_0=1+y$, already the normalized root. A
subsequent replacement by $R_0(y^p)=1+y^p$ would have $p$-th power
$1+y^{p^2}$, not $H$.

The displayed expression
\[
  \left(1+\frac{yF_1}{\alpha}\right)^{1/e}
\]
is best read as the normalized unit root in the Laurent ring
$F[x_1^{\pm 1},\ldots,x_n^{\pm 1}][[y]]$. The actual polynomial input to the
binomial expansion is
\[
  H=\alpha^{-1}F^\sharp(x,y\alpha).
\]
Since $F^\sharp=\alpha+yF_1$, the $y$-degree zero part of $H$ is $1$. For a
term $My^j$ of $F^\sharp$ with $j\geq 1$, the transformation sends it to
$M\alpha^{j-1}y^j$, which is a polynomial term because $\alpha$ is a monomial
times a nonzero scalar. Thus $H$ is reverse-monic in $y$, and the Algorithm
2-style expansion is applied to an ordinary polynomial. The inverse substitution
$y\mapsto y/\alpha$ may temporarily produce Laurent monomials; therefore the
algorithm explicitly checks that $G^\sharp_{\mathrm{cand}}$ is an ordinary
polynomial before specializing $y=1$.

The construction of $F^\sharp$ is polynomial because $k_f$ is the minimum
$k$-weight among the monomials of $f$, so every exponent
$\langle \mathbf{u},\mathbf{k}\rangle-k_f$ appearing after the substitution is nonnegative. The
same homomorphism calculation as in Theorem~\ref{thm:unified_sep} gives
$F^\sharp=(G^\sharp)^e$ whenever $f=g^e$, and specialization at $y=1$ gives
\[
  G^\sharp(x_1,\dots,x_n,1)=g(x_1,\dots,x_n).
\]
Thus the scalar transformation preserves exact powers and is inverted on the
base by setting $y=1$.

The root extraction of $\alpha=c\mathbf{x}^{\mathbf{a}}$ is necessary over $\F$. Indeed, from
$F^\sharp=(G^\sharp)^e$ and $F^\sharp(x_1,\dots,x_n,0)=c\mathbf{x}^{\mathbf{a}}$, the specialization
$G^\sharp(x_1,\dots,x_n,0)$ has $e$-th power $c\mathbf{x}^{\mathbf{a}}$. In a polynomial ring over a field, a
polynomial whose positive power is a single nonzero monomial must itself be a
single nonzero monomial. Hence $G^\sharp(x_1,\dots,x_n,0)=\beta \mathbf{x}^{\mathbf{b}}$ with $\beta^e=c$ and
$e\mathbf{b}=\mathbf{a}$. Failure of either condition rules out a base over $\F$ with the
prescribed transformed lowest term.

After pseudo-monic normalization, $H=R^e$ with $R(x_1,\dots,x_n,0)=1$. If the
characteristic does not divide the active exponent, the coefficients
$\binom{1/e}{i}$ are defined in $\F$, and the $y$-adic binomial expansion gives
the unique root with constant term $1$ modulo $y^{N+1}$. The truncation level
$N=\lfloor\deg_y(F^\sharp)/E\rfloor$ is sufficient because
$\deg_y(G^\sharp)=\deg_y(F^\sharp)/e$ for a genuine exact power, and the
normalization preserves $y$-degree. In the characteristic-dividing case, the
true Frobenius-root reduction replaces $H=R^{p^rq}$ by $H_0=R^q$, so the same
binomial argument with active exponent $q$ recovers $R$ itself. No additional
$p^r$-power substitution is applied before inverse normalization.

Finally, all possible false candidates are rejected by the last check. The
normalization, the Laurent inverse substitution, or an extension-field scalar
root can introduce an expression that is not a polynomial over $F$, and the
algorithm rejects it. Even if a polynomial candidate over $\F$ is produced, the
algorithm returns it only after the exact multiplication check
$g_{\mathrm{cand}}^e=f$ in $\F[x_1,\ldots,x_n]$. Therefore Algorithm \ref{alg:phase2-base-reconstruction} never outputs
a base outside the original field and never accepts an incorrect base.
The detailed field-operation accounting, including scalar-root construction,
the cost of Frobenius-root checks, inverse normalization, specialization, and
final exact verification, remains deferred to Task 6.

\begin{proposition}[Correctness of the deterministic exact-root algorithm]
\label{prop:exact-root-correctness}
Let $f\in \F[x_1,\ldots,x_n]\setminus \F$ and let $e\ge 2$ be an integer.
The reconstruction procedure in Algorithm~\ref{alg:phase2-base-reconstruction}
has the following correctness property.

\begin{enumerate}
\item If there exists $g\in \F[x_1,\ldots,x_n]$ such that $f=g^e$, then the
algorithm outputs a polynomial $\tilde g\in \F[x_1,\ldots,x_n]$ satisfying
$\tilde g^e=f$.

\item If no such $g\in \F[x_1,\ldots,x_n]$ exists, then the algorithm
outputs ``not an $e$-th power''.
\end{enumerate}
\end{proposition}

\begin{proof}
We prove the reconstruction step for the given exponent $e\ge 2$. Suppose $f=g^e$ for some $g\in \F[x_1,\ldots,x_n]$. Apply Theorem~\ref{thm:unified_sep} to $f=g^e$. It supplies a separating vector $\mathbf{k}=(k_1,\ldots,k_n)$, a unique lowest term $c\mathbf{x}^{\mathbf{a}}$ of $f$, and an integer $k_f=\langle \mathbf{a},\mathbf{k}\rangle$ such that
\[
  F^\sharp(\mathbf{x},y)=y^{-k_f}f(x_1y^{k_1},\ldots,x_ny^{k_n})
\]
is an ordinary polynomial. Indeed, for a term $c_{\mathbf{u}}\mathbf{x}^{\mathbf{u}}$ of $f$, the resulting $y$-exponent is $\langle \mathbf{u},\mathbf{k}\rangle-k_f$, which is nonnegative by the defining minimality of $k_f$. Theorem~\ref{thm:unified_sep} also gives
\[
  G^\sharp(\mathbf{x},y)=y^{-k_g}g(x_1y^{k_1},\ldots,x_ny^{k_n})
  \in \F[x_1,\ldots,x_n,y],
\]
with $k_f=e k_g$, and the ring-homomorphism calculation gives
\[
  F^\sharp=(G^\sharp)^e,\qquad F^\sharp(\mathbf{x},0)=c\mathbf{x}^{\mathbf{a}}.
\]
Specializing the displayed definition of $G^\sharp$ at $y=1$ gives
\[
  G^\sharp(\mathbf{x},1)=g(\mathbf{x}).
\]
Thus the scalar transformation preserves exact powers and recovers the original base by setting $y=1$ after the transformed base has been found.

We next justify the initial root extraction. Since
\[
  F^\sharp(\mathbf{x},0)=G^\sharp(\mathbf{x},0)^e=c\mathbf{x}^{\mathbf{a}},
\]
the polynomial $G^\sharp(\mathbf{x},0)$ has a positive power equal to one nonzero monomial. A polynomial with at least two distinct monomials cannot have a positive power with only one monomial: for any separating weight on its support, the lowest and highest weight monomials in the power are distinct and cannot cancel. Hence
\[
  G^\sharp(\mathbf{x},0)=\beta \mathbf{x}^{\mathbf{b}}
\]
for some $\beta\in \F^*$ and $\mathbf{b}\in\mathbb N^n$. Comparing the single monomial on both sides gives
\[
  \beta^e=c,\qquad e\mathbf{b}=\mathbf{a}.
\]
Therefore an $e$-th root of $c\mathbf{x}^{\mathbf{a}}$ over $\F[\mathbf{x}]$ is necessary. It is also sufficient for the normalization step: once $\beta$ and $\mathbf{b}$ satisfying these equations are chosen, $B=\beta \mathbf{x}^{\mathbf{b}}$ satisfies $B^e=c\mathbf{x}^{\mathbf{a}}$.

This explains the field-root language in the algorithm. If $c$ has no $e$-th root in $\F$, or if some coordinate of $\mathbf{a}$ is not divisible by $e$, then no base over $\F[x_1,\ldots,x_n]$ can produce the displayed initial monomial. If a root is computed in an extension field, it is used only to generate candidates. Such a candidate is not accepted as an output over $\F$ unless, after inverse normalization and specialization, all coefficients lie in $\F$ and the equality $g_{\mathrm{cand}}^e=f$ holds in $\F[x_1,\ldots,x_n]$.

Put
\[
  \alpha=c\mathbf{x}^{\mathbf{a}}=B^e,\qquad F^\sharp=\alpha+yF_1.
\]
The direct expression $(1+yF_1/\alpha)^{1/e}$ may live in a Laurent ring in the $\mathbf{x}$-variables, so the algorithm uses the polynomial pseudo-monic normalization
\[
  H=\alpha^{-1}F^\sharp(x_1,\ldots,x_n,y\alpha).
\]
For the $y$-degree zero term $\alpha$ of $F^\sharp$, this transformation gives $1$. For a term $M(\mathbf{x})y^j$ of $F^\sharp$ with $j\ge 1$, substitution $y\mapsto y\alpha$ gives $M(\mathbf{x})\alpha^j y^j$, and division by $\alpha$ gives $M(\mathbf{x})\alpha^{j-1}y^j$, an element of $\F[x_1,\ldots,x_n,y]$ because $\alpha$ is a nonzero scalar times a monomial. Thus
\[
  H\in \F[x_1,\ldots,x_n,y],\qquad H(\mathbf{x},0)=1.
\]
Now define
\[
  R=B^{-1}G^\sharp(x_1,\ldots,x_n,y\alpha).
\]
The same monomial calculation, applied to $G^\sharp$ and $B$, shows that $R$ is an ordinary polynomial and $R(\mathbf{x},0)=1$. Moreover,
\[
  H
  =\alpha^{-1}F^\sharp(\mathbf{x},y\alpha)
  =B^{-e}G^\sharp(\mathbf{x},y\alpha)^e
  =R^e.
\]
Thus $H$ is reverse-monic in the variable $y$ and is the $e$-th power of the normalized root $R$.

We now prove the binomial step. Refer to Corollary \ref{binomialExactPower}.
Since $F^\sharp=(G^\sharp)^e$, one has $\deg_y(F^\sharp)=e\deg_y(G^\sharp)$. The substitution $y\mapsto y\alpha$ and multiplication by $B^{-1}$ do not change $y$-degrees, so
\[
  \deg_y(R)=\deg_y(G^\sharp)=\deg_y(F^\sharp)/e\le N.
\]
Thus equality modulo $y^{N+1}$ determines the entire polynomial $R$. The binomial truncation therefore returns $R$, not merely an approximation.

It remains to handle the case in which $\operatorname{char}(\F)=p>0$ divides $e$. Write
\[
  e=p^r q,\qquad p\nmid q.
\]
Before the binomial expansion, the algorithm takes a true inverse Frobenius root of $H$: it divides all monomial exponent vectors, including the exponent of $y$, by $p^r$, and it takes the corresponding $p^r$-th roots of coefficients in $\F$. In the promised case this operation is defined, because
\[
  H=R^e=R^{p^r q}=(R^q)^{p^r}
\]
in characteristic $p$. The Frobenius map on a field is injective, so when the required coefficient roots lie in $\F$ they are unique. The result of the true inverse Frobenius reduction is
\[
  H_0=R^q.
\]
Now $q$ is invertible in $\F$, so the preceding uniqueness and truncation argument applies to $H_0$ with active exponent $q$. It recovers the same normalized polynomial $R$. No subsequent substitution
\[
  R(\mathbf{x},y)\mapsto R(x_1^{p^r},\ldots,x_n^{p^r},y^{p^r})
\]
is applied, because the inverse Frobenius reduction has already converted $H=R^{p^r q}$ into $H_0=R^q$.

After the normalized root $R$ is computed, the algorithm forms
\[
  G^\sharp_{\mathrm{cand}}=B\,R(x_1,\ldots,x_n,y/\alpha).
\]
For the true normalized root constructed above, this is exactly $G^\sharp$, since the definition
\[
  R=B^{-1}G^\sharp(\mathbf{x},y\alpha)
\]
is inverted by substituting $y/\alpha$ for $y$ and multiplying by $B$. Hence $G^\sharp_{\mathrm{cand}}\in \F[x_1,\ldots,x_n,y]$ in the promised exact-power case, and
\[
  g_{\mathrm{cand}}=G^\sharp_{\mathrm{cand}}(\mathbf{x},1)=G^\sharp(\mathbf{x},1)=g.
\]
The final multiplication check then accepts, because $g_{\mathrm{cand}}^e=f$.

For arbitrary intermediate candidates, the last steps are also sound. If inverse normalization produces a Laurent expression or a polynomial with coefficients outside $\F$, the algorithm rejects it and therefore outputs no invalid base over $\F$. If it produces a polynomial $g_{\mathrm{cand}}\in \F[x_1,\ldots,x_n]$, the algorithm accepts only after checking the exact identity
\[
  g_{\mathrm{cand}}^e=f
\]
in the original polynomial ring. Therefore every accepted output is correct.

Combining this reconstruction result with the degree bound gives the proposition. Under the exact-power promise, the genuine transformed root satisfies all necessary scalar, monomial, Frobenius, polynomialhood, and final-equality checks, so the algorithm has a passing candidate. If an implementation using extension-field candidate generation or incomplete scalar-root search produces no passing candidate, the safe conclusion is reconstruction failure for that exponent, not a proof that $f$ is not an $e$-th power and not acceptance of an incorrect base. In the exact algorithm over the stated field model, the genuine candidate passes the final check.
\end{proof}

\begin{theorem}[Complexity of the deterministic exact-root algorithm]
\label{thm:alg6}
Let
\[
  f\in \F[x_1,\ldots,x_n]
\]
be given in sparse representation, and let $e\ge 2$ be an integer. Let
$s=\|f\|_0$,
\[
  d=\operatorname{indeg}(f),\qquad D=\tdeg(f).
\]
Then the algorithm computes a polynomial $g\in \F[x_1,\ldots,x_n]$ such that
$g^e=f$, or determines that no such $g$ exists, in time
\[
  \operatorname{poly}\bigl(s^{O(Dd)},n,d,D\bigr)+s\cdot R(e),
\]
where $R(e)$ denotes the cost of constructing a single $e$-th root of a scalar in the base field $\F$, and, when $\operatorname{char}(\F)\mid e$, the cost of computing a single Frobenius root of a scalar (including both existence testing and construction).
\end{theorem}

\begin{proof}
For the scalar transform, Theorem~\ref{thm:unified_sep} supplies a separating vector
$\mathbf{k}\in\{1,\ldots,2d+2\}^n$ and forms
\[
  F^\sharp=y^{-k_f}f(x_1y^{k_1},\ldots,x_ny^{k_n}).
\]
Every monomial of $f$ has total degree at most $D$, so its $y$-exponent after
substitution is bounded above by $D(2d+2)$. Hence
\[
  \Delta_y:=\deg_y(F^\sharp)\le D(2d+2).
\]
In an exact $e$-th-power case, the transformed root $G^\sharp$ satisfies
$F^\sharp=(G^\sharp)^e$, and the structural sparsity bound gives
\[
  \|G^\sharp\|_0=\|g\|_0\le s^{D(2d+2)/e+1}.
\]

The pseudo-monic normalization $H=\alpha^{-1}F^\sharp(\mathbf{x},y\alpha)$
does not increase the $y$-degree. The binomial expansion is computed only up
to $N=\lfloor \Delta_y/e\rfloor$, because in a genuine exact power the
normalized root has $y$-degree $\Delta_y/e$. Under the stated sparse
arithmetic model, each monomial substitution, truncation, multiplication,
Laurent membership check, inverse normalization, and specialization is
polynomial in the number of sparse terms produced and in the degree
parameters. Since the produced root has sparsity at most
$s^{D(2d+2)/e+1}$, these operations are bounded by
\[
  \operatorname{poly}\bigl(s^{O(Dd/e)},n,d,D\bigr)
\]
apart from scalar-field operations.

The reconstruction phase requires $O(s)$ scalar root operations: one $e$-th root of the lowest scalar coefficient of $f$, and in positive characteristic with $\operatorname{char}(\F)\mid e$, up to $O(s)$ Frobenius root operations for the coefficients of the transformed polynomial $H$. Each such operation costs $R(e)$ as defined above. Hence the scalar-field contribution is $s\cdot R(e)$.

It remains to count final verification. The algorithm does not accept the
candidate unless it checks the exact identity $g_{\rm cand}^{\,e}=f$. Using
ordinary sparse powering with $\|g_{\rm cand}\|_0\le s^{D(2d+2)/e+1}$, the
intermediate sparse-size is bounded by
\[
  \left(s^{D(2d+2)/e+1}\right)^e
  =s^{D(2d+2)+e}.
\]
Since $e\le d\le D$ for nonconstant inputs, this is absorbed by
$\operatorname{poly}(s^{O(Dd)},n,d,D)$. Combining this verification cost with
the candidate-construction cost gives the total bound
\[
  \operatorname{poly}\bigl(s^{O(Dd)},n,d,D\bigr)+s\cdot R(e).
\]

If a verification oracle is assumed, its cost replaces the ordinary sparse-
powering term, yielding the conditional $s^{O(Dd/e)}$-scale bound.

The scalar term $R(e)$ packages two types of operations: scalar $e$-th root construction and, in positive characteristic, scalar Frobenius-root construction. Over finite fields, these costs are polynomial if deterministic root-construction routines are available. Over $\mathbb Q$, they reduce to integer/rational exact-root arithmetic. Over general number fields, the absorption requires a specified representation and a deterministic scalar-root algorithm. Thus $R(e)$ is absorbed into the polynomial term only in settings where the relevant scalar operations are themselves available with polynomial cost.
\end{proof}

\section{Conclusion}

In this paper we have studied the problem of deterministically computing exact roots of sparse polynomials. For a  polynomial $f\in \F[x_1,\ldots,x_n]$ with sparsity $s=\|f\|_0$, individual degree $d=\operatorname{indeg}(f)$, and total degree $D=\tdeg(f)$, we proved that if $f=g^e$ for some $g\in \F[x_1,\ldots,x_n]$ and $e\ge 2$, then
\[
  \|g\|_0 \le s^{D(2d+2)/e+1}.
\]
This sparsity bound is derived via a scalar pseudo-monic transformation followed by the pseudo-monic exact-root lemma. Importantly, the exponent depends only on $D$, $d$, and $e$, and not on the number of variables $n$, making it suitable for polynomial-time algorithms in the bounded-degree regime.

Based on this structural bound, we developed a deterministic algorithm that, given $f$ and an integer $e\ge 2$, decides whether $f=g^e$ for some $g\in \F[x_1,\ldots,x_n]$ and, if so, outputs such a $g$. The algorithm proceeds by a scalar weighted substitution to reduce the problem to the pseudo-monic setting, followed by a $y$-adic binomial expansion with Frobenius reduction in positive characteristic, inverse normalization, and final exact verification. The overall complexity is
\[
  \operatorname{poly}\bigl(s^{O(Dd)},n,d,D\bigr)+s\cdot R(e),
\]
where $R(e)$ denotes the cost of constructing a single $e$-th root of a scalar in the base field $\F$, and, when $\operatorname{char}(\F)\mid e$, the cost of computing a single Frobenius root of a scalar. When $d$ and $D$ are bounded, this yields a deterministic polynomial-time algorithm for exact-root computation, in contrast to the quasi-polynomial dependence of the general deterministic factorization algorithm of Bhargava, Saraf, and Volkovich.

In addition, we established a sparsity bound for multilinear cofactors: if $f=gh$ with $h$ multilinear, then $\|h\|_0\le s$ and $\|g\|_0\le s^{(2d+2)D+2}$. These bounds, while not the main focus of this work, provide structural information complementary to our exact-root results.

Several directions remain for future work. A natural next step is to extend the exact-root sparsity bound and algorithm to the general factorization setting $f=gh$, where both factors are arbitrary sparse polynomials with bounded degree. The current multilinear-cofactor bound represents a first step in this direction, but a complete deterministic factorization algorithm with polynomial complexity would require a polynomial-size factor-sparsity bound for all factors, or a more refined algorithmic framework that avoids such a bound. Another important direction is to eliminate the scalar-field cost $R(e)$ by developing deterministic polynomial-time scalar root algorithms for arbitrary fields, or to show that such costs are unavoidable. Finally, extending the exact-root algorithm to output a full factorization of $f$ into irreducible factors remains a challenging open problem.

\appendix

\section{Detailed Proofs}\label{app:detailed-proofs}

This appendix gives the detailed proof bookkeeping for the main statements. It is a reader-facing expansion of the proof package; it does not replace the theorem statements in the body and does not prove the stronger variants discussed above.

\subsection{Detailed proof of Theorem~\ref{thm:exact}}

We prove the exact-power theorem under the nonconstant hypotheses of
Theorem~\ref{thm:exact}. Since $f\notin F$, the polynomial $f$ is nonzero and
nonconstant. Thus $\supp(f)$ is finite and nonempty. Moreover
$\indeg(f)\geq 1$, so the hypothesis $\indeg(f)\leq d$ implies $d\geq 1$.

Let $A=\supp(f)$. Then $A$ is finite, nonempty, and contained in $\{0,\ldots,d\}^n$. The separating-vector lemma from the constant-free exact-root note gives a vector
\[
  k=(k_1,\ldots,k_n)\in \{1,\ldots,2d+2\}^n
\]
and an exponent $a_0\in A$ such that
\[
  \langle a_0,k\rangle < \langle a,k\rangle
  \qquad(a\in A,\ a\neq a_0).
\]
Writing the corresponding lowest monomial of $f$ as $c_0x^{a_0}$, with $c_0\in F^*$, the lowest-root-monomial lemma applies to $f=g^e$. It gives a unique lowest $k$-weight monomial $\gamma x^{b_0}$ of $g$ such that
\[
  \gamma\in F^*,\qquad a_0=eb_0,\qquad c_0=\gamma^e,\qquad
  \langle a_0,k\rangle=e\langle b_0,k\rangle.
\]
Set
\[
  k_f=\langle a_0,k\rangle,\qquad k_g=\langle b_0,k\rangle.
\]
Then $k_f=ek_g$. Define the scalar transformed polynomials
\[
  F^\sharp=y^{-k_f}f(x_1y^{k_1},\ldots,x_ny^{k_n}),\qquad
  G^\sharp=y^{-k_g}g(x_1y^{k_1},\ldots,x_ny^{k_n}).
\]
For a monomial $c_ax^a$ of $f$, the corresponding monomial in $F^\sharp$ is
\[
  c_ax^a y^{\langle a,k\rangle-k_f}.
\]
Because $k_f$ is the minimum $k$-weight on $\supp(f)$, all exponents of $y$ here are nonnegative. Hence $F^\sharp\in F[x_1,\ldots,x_n,y]$. The same argument with $k_g$ gives $G^\sharp\in F[x_1,\ldots,x_n,y]$.

The scalar substitution is a ring homomorphism. Therefore
\[
  F^\sharp
  =
  y^{-k_f}f(x_1y^{k_1},\ldots,x_ny^{k_n})
  =
  y^{-ek_g}g(x_1y^{k_1},\ldots,x_ny^{k_n})^e
  =
  (G^\sharp)^e.
\]
Specializing at $y=0$ keeps exactly the monomials of minimum $k$-weight. Since $a_0$ is the unique minimum exponent of $f$,
\[
  F^\sharp|_{y=0}=c_0x^{a_0},
\]
a single nonzero monomial. Thus $F^\sharp$ is $\{y\}$-reverse pseudo-monic.

The support comparison is direct. The transformation sends an exponent $a\in\supp(f)$ to
\[
  (a,\langle a,k\rangle-k_f)\in\mathbb N^{n+1}.
\]
This map is injective because its first $n$ coordinates recover $a$, and no nonzero coefficient is killed. Hence
\[
  \|F^\sharp\|_0=\|f\|_0\leq s.
\]
Similarly, $b\mapsto (b,\langle b,k\rangle-k_g)$ gives
\[
  \|G^\sharp\|_0=\|g\|_0.
\]

The scalar-substitution degree estimate gives the revised $y$-degree parameter. Write
\[
  f=c_1m_1+\cdots+c_tm_t,
  \qquad
  m_i=x_1^{e_1}\cdots x_n^{e_n}
\]
for a fixed monomial $m_i$ in the support of $f$. Under
\[
  x_j\mapsto x_jy^{k_j}\qquad(1\leq j\leq n),
\]
this monomial becomes
\[
  (x_1y^{k_1})^{e_1}\cdots (x_ny^{k_n})^{e_n}
  =
  x_1^{e_1}\cdots x_n^{e_n}
  y^{e_1k_1+\cdots+e_nk_n}.
\]
Hence its $y$-degree after the scalar substitution is
\[
  e_1k_1+\cdots+e_nk_n
  \leq (e_1+\cdots+e_n)(2d+2)
  \leq D(2d+2),
\]
because $k_j\leq 2d+2$ for every $j$ and every monomial of $f$ has total degree at most $D=\tdeg(f)$. Taking the maximum over the monomials of $f$ gives
\[
  \Dlin\leq D(2d+2).
\]
Since the exponent of $y$ in $F^\sharp$ is $\langle a,k\rangle-k_f\leq \langle a,k\rangle$, the same calculation gives
\[
  \tdeg_y(F^\sharp)\leq D(2d+2).
\]

We now apply the local pseudo-monic exact-root lemma to
\[
  H=F^\sharp,\qquad R=G^\sharp,\qquad S=s,\qquad \Delta=D(2d+2)
\]
in the polynomial ring $F[x_1,\ldots,x_n,y]$. Its preconditions are checked above: $F^\sharp$ is $s$-sparse, $\{y\}$-reverse pseudo-monic, has $y$-degree at most $D(2d+2)$, and is the $e$-th power of $G^\sharp$. Hence
\[
  \|G^\sharp\|_0\leq s^{D(2d+2)/e+1}.
\]

\subsection{Detailed proof of Theorem~\ref{thm:multi}}

The factorization is oriented: $f=gh$, the named factor $h$ is multilinear, and the product is nonzero. The nonzero convention is necessary for the factor-degree bridge and avoids false zero-product variants such as $g=0$ with arbitrary multilinear $h$.

First, Lemma 2.2 of Bisht--Volkovich applies to the named multilinear divisor $h$ of $f$. It gives
\[
  \|h\|_0\leq \|f\|_0\leq s.
\]
This is the multilinear orientation point: the theorem bounds the named multilinear factor $h$ directly, and the other factor $g$ is bounded by the cofactor theorem.

The degree hypothesis is $\indeg(f)\leq d$. Fix a variable $x_i$ and view
\[
  F[x_1,\ldots,x_n]=R[x_i],
  \qquad
  R=F[x_1,\ldots,x_{i-1},x_{i+1},\ldots,x_n].
\]
The coefficient ring $R$ is an integral domain. Since $f=gh$ is a nonzero product, degree in $x_i$ is additive:
\[
  \tdeg_{x_i}(f)=\tdeg_{x_i}(g)+\tdeg_{x_i}(h).
\]
Consequently $\tdeg_{x_i}(g)\leq \tdeg_{x_i}(f)\leq d$. This holds for each $i$, so $\indeg(g)\leq d$. Total degree is also additive for nonzero products, so $\tdeg(g)\leq\tdeg(f)=D$.

Since $\indeg(f)\leq d$, the support of $f$ lies in $\{0,\ldots,d\}^n$. Apply the bounded separating-vector lemma used in Section~4 to obtain
\[
  a_0\in\supp(f),\qquad k=(k_1,\ldots,k_n)\in\{1,\ldots,2d+2\}^n
\]
such that $a_0$ uniquely minimizes $w(a)=\langle a,k\rangle$ on $\supp(f)$. Write the corresponding term of $f$ as $c_0x^{a_0}$ with $c_0\in F^*$, and set $k_f=w(a_0)$. Define
\[
  F^\sharp=y^{-k_f}f(x_1y^{k_1},\ldots,x_ny^{k_n}).
\]
For a term $c_ax^a$ of $f$, its image in $F^\sharp$ is $c_ax^ay^{w(a)-k_f}$. The exponent $w(a)-k_f$ is nonnegative by minimality of $k_f$, hence $F^\sharp\in F[x_1,\ldots,x_n,y]$. The specialization at $y=0$ keeps exactly the unique lowest-weight term:
\[
  F^\sharp(x_1,\ldots,x_n,0)=c_0x^{a_0}.
\]

For a nonzero polynomial $P$, let
\[
  \mu(P)=\min\{\langle a,k\rangle:x^a\in\supp(P)\}.
\]
Set $k_g=\mu(g)$ and $k_h=\mu(h)$, and define
\[
  G^\sharp=y^{-k_g}g(x_1y^{k_1},\ldots,x_ny^{k_n}),\qquad
  H^\sharp=y^{-k_h}h(x_1y^{k_1},\ldots,x_ny^{k_n}).
\]
The same minimum-weight argument shows that $G^\sharp,H^\sharp\in F[x_1,\ldots,x_n,y]$. The substitution $P\mapsto P(x_1y^{k_1},\ldots,x_ny^{k_n})$ is a ring homomorphism, and minimum weights in the nonzero product satisfy $k_f=k_g+k_h$: every product term has weight at least $k_g+k_h$, and products of lowest-weight terms have that weight. Therefore
\[
  F^\sharp
  =
  \bigl(y^{-k_g}g(x_1y^{k_1},\ldots,x_ny^{k_n})\bigr)
  \bigl(y^{-k_h}h(x_1y^{k_1},\ldots,x_ny^{k_n})\bigr)
  =
  G^\sharp H^\sharp.
\]

Let $\operatorname{in}_k(P)$ denote the sum of the terms of $P$ with minimal $k$-weight. The equality $F^\sharp|_{y=0}=c_0x^{a_0}$ says that $\operatorname{in}_k(f)$ is a single monomial. Since $f=gh$,
\[
  \operatorname{in}_k(f)=\operatorname{in}_k(g)\operatorname{in}_k(h).
\]
In a polynomial ring over a field, if the product of two nonzero polynomials is a monomial, then each factor is a monomial. Indeed, with respect to any monomial order, a polynomial with more than one term has distinct least and greatest monomials; the least and greatest monomials of a product are the products of the corresponding least and greatest monomials. Thus $\operatorname{in}_k(h)$ is a nonzero monomial. Since
\[
  H^\sharp(x_1,\ldots,x_n,0)=\operatorname{in}_k(h),
\]
the transformed named multilinear factor $H^\sharp$ is $\{y\}$-reverse pseudo-monic.

The scalar substitution preserves supports. For each $P\in\{f,g,h\}$, a monomial $c_ax^a$ maps to $c_ax^ay^{\langle a,k\rangle-\mu(P)}$, and the exponent map $a\mapsto(a,\langle a,k\rangle-\mu(P))$ is injective. Hence
\[
  \|F^\sharp\|_0=\|f\|_0,\qquad
  \|G^\sharp\|_0=\|g\|_0,\qquad
  \|H^\sharp\|_0=\|h\|_0.
\]
In particular, $\|F^\sharp\|_0\leq s$ and $\|H^\sharp\|_0\leq s$.

It remains to bound the relevant individual degree for the cofactor estimate. If $c_bx^b$ is a term of $g$, then its $y$-exponent in $G^\sharp$ is $\langle b,k\rangle-k_g$. Since $k_g\geq0$ and $k_i\leq2d+2$ for all $i$,
\[
  \langle b,k\rangle-k_g
  \leq
  \langle b,k\rangle
  \leq
  (2d+2)\sum_i b_i
  \leq
  (2d+2)D.
\]
Therefore $\tdeg_y(G^\sharp)\leq(2d+2)D$.

Apply Corollary~\ref{cor:pseudo}, the $I$-reverse pseudo-monic cofactor result, to
\[
  F^\sharp=G^\sharp H^\sharp
\]
with $I=\{y\}$ and $\delta=(2d+2)D$. Its preconditions have been checked: $H^\sharp$ is $\{y\}$-reverse pseudo-monic, $F^\sharp$ and $H^\sharp$ are both $s$-sparse, and the individual degree of $G^\sharp$ in $y$ is at most $\delta$. Thus
\[
  \|G^\sharp\|_0\leq s^{(2d+2)D+2}.
\]
Since $\|G^\sharp\|_0=\|g\|_0$, this gives
\[
  \|g\|_0\leq s^{(2d+2)D+2}.
\]

\subsection{Detailed proof of Theorem~\ref{thm:alg6}}

Theorem~\ref{thm:alg6} is a promised algorithmic statement, not an arbitrary exact-power test. Its input polynomial $H\in F[z_1,\ldots,z_m]$ must already satisfy the following conditions:
\[
  \|H\|_0\leq S,\qquad \indeg(H)\leq \Delta,\qquad H|_{z_j=0}=1
\]
for some variable $z_j$, and it must be promised that $H=R^e$ for a polynomial $R$. Under exactly these hypotheses, Bisht--Volkovich Lemma 5.8 and Algorithm 2 compute $R$ in
\[
  \poly(S^{\Delta/e},m,\Delta)
\]
field operations.

The scalar exact-power transform from Section~4 supplies the pseudo-monic data
\[
  H=F^\sharp,\qquad R=G^\sharp,\qquad m=n+1,\qquad S=s,\qquad I=\{y\},\qquad \Delta=D(2d+2).
\]
Under the nonconstant hypotheses of Theorem~\ref{thm:exact}, the separating vector, lowest-root monomial, scalar transformation definitions, support comparison, and degree bound put the transformed input into $\{y\}$-reverse pseudo-monic form. The Phase 2 reconstruction in Section~6 normalizes this pseudo-monic input to a reverse-monic polynomial in $y$, applies the binomial-root step after the characteristic-dividing Frobenius reduction when needed, undoes the normalization, specializes $y=1$, and accepts only after an exact verification in $F[x_1,\ldots,x_n]$. The same transformed setup also gives the structural term bound
\[
  \|G^\sharp\|_0=\|g\|_0\leq s^{D(2d+2)/e+1}.
\]

\bibliographystyle{klmm}
\bibliography{refs}

\end{document}